\documentclass{article}
\usepackage{rompmod}

\usepackage{dcolumn}
\usepackage{bm}

\usepackage{amsthm}
\usepackage{amssymb}

\usepackage{amstext}
\usepackage{amsmath}
\usepackage{amsfonts}
\usepackage{units}
\usepackage{mathrsfs}
\usepackage{cite}
\newcommand{\p}{\partial}

\newcommand{\dd}{{\rm d}}

\theoremstyle{theorem}
\newtheorem{theorem}{Theorem}
\theoremstyle{remark}
\newtheorem{remark}{Remark}

\theoremstyle{theorem}
\newtheorem{proposition}{Proposition}

\theoremstyle{lemma}
\newtheorem{lemma}{Lemma}

\theoremstyle{definition}

\begin{document}

\title{An equivalence of Finslerian relativistic theories}

\author{E. Minguzzi\thanks{
Dipartimento di Matematica e Informatica ``U. Dini'', Universit\`a
degli Studi di Firenze, Via S. Marta 3,  I-50139 Firenze, Italy.
 e-mail: ettore.minguzzi@unifi.it }}

\date{}

\maketitle

\begin{abstract}
\noindent 
In Lorentz-Finsler geometry it is natural to define
the Finsler Lagrangian over a   cone (Asanov's approach) or over the whole slit tangent bundle (Beem's approach). In the former case one might want to add differentiability conditions at the boundary of the (timelike) cone in order to retain the usual definition of lightlike geodesics. It is shown here that if this is done then the two theories  coincide, namely the `conic' Finsler Lagrangian is the restriction of a slit tangent bundle Lagrangian. Since causality theory depends on curves defined through the future cone, this work establishes the essential uniqueness of (sufficiently regular) Finsler spacetime theories and Finsler causality.
\end{abstract}


\section{Introduction}

 In this work we shall establish that any sufficiently regular Lorentz-Finsler Lagrangian defined over a cone $\mathscr{L}\le 0$ can be extended to the whole slit tangent bundle (Theorem \ref{odd}). This result is expected to be useful in both Lorentz-Finsler geometry and in Finslerian theories of gravity since it implies the equivalence of apparently different approaches. I am referring to the works who followed  Asanov's  conic approach \cite{asanov81,asanov85,pimenov88,bejancu00,barletta12,javaloyes13} and to the works who followed  Beem's slit tangent bundle approach \cite{beem70,beem76b,akbarzadeh88,bejancu99,perlick06,skakala09,minguzzi13c,minguzzi14c,voicu15}, and which can be considered, thanks to this work, as studies of the same theory. Some Finslerian approaches which cannot be easily comprised in this unified theory include \cite{pfeifer11,kostelecky11,lammerzahl12}.

  The mathematical methods  used in this work belong to Finsler geometry of indicatrices, affine differential geometry and analysis of convex functions \cite{li93,nomizu94,borwein10}. The reader is referred to \cite{minguzzi13c} for a more physically oriented bibliography on Finsler spacetime theories including proposals that have been advanced for the field equations.

Let $M$ be a paracompact, Hausdorff, connected, $n+1$-dimensional manifold. Let $\{x^\mu\}$ denote a local chart on $M$ and let $\{ x^\mu,v^\nu\}$ be the induced local chart on $TM$.
We start giving a quite general setting for Finsler spacetime theory, which we call the  {\em rough model}.

Let $\Omega$ be a  subbundle of  the slit tangent bundle, $\Omega \subset TM\backslash 0$, such that $\Omega_x$ is an open sharp convex cone for every $x$.
A {\em Finsler Lagrangian} is a map  $\mathscr{L}\colon \Omega \to \mathbb{R}$ which
is  positive homogeneous of degree two in the fiber coordinates
\[
\mathscr{L}(x,sv)=s^2 \mathscr{L}(x,v), \qquad \forall s>0.
  \]
  It is assumed that the fiber dependence  is $C^2(\Omega)$,
 that $\mathscr{L}<0$ on $\Omega$ and that $\mathscr{L}$ can be continuously extended setting $\mathscr{L}=0$ on $\p \Omega$.
The metric is defined as the Hessian  of $\mathscr{L}$ with respect to the fibers
\[
g_{\mu \nu}(x,v)= \frac{\p^2 \mathscr{L}}{\p v^\mu \p v^\nu},
\]
and in index free notation will be denoted with $g_v$ to stress the dependence on the velocity. This Finsler metric provides a map $g\colon \Omega \to  T^*M \otimes T^*M$. The manifold $(M,\mathscr{L})$ is called a {\em Finsler spacetime} whenever   $g_v$ is Lorentzian, namely of signature $(-,+,\cdots,+)$. By positive homogeneity we have $\mathscr{L}=\frac{1}{2} g_v(v,v)$ and $\dd L=g_v(v,\cdot)$. The usual Lorentzian-Riemannian case is obtained for $\mathscr{L}$ quadratic in the velocities.

A vector $v\in \bar \Omega$ is said to be future timelike, lightlike, or spacelike depending on the sign of $\mathscr{L}(x,v)$, respectively negative, zero, or positive. We  denote the sets of these vectors with $I^+=\Omega$, $E^+=\p \Omega$ and $J^+=\bar{\Omega}$, respectively. The observer space (indicatrix), or velocity space, is $\mathscr{I}^-_x=\{v\in T_xM\colon 2\mathscr{L}(x,v)=-1\}$. The condition $\mathscr{L}(x,v) \to 0$ for $v\to \p \Omega$ assures that the observer space is fully contained in $I^+_x$. In particular, it is non-compact and asymptotic to $E^+_x$.
Observe that the Finsler Lagrangian is defined just over a subset of the slit tangent bundle
as pioneered by Asanov \cite{asanov85}.

Beem's definition of Finsler spacetime is more demanding \cite{beem70}, as in his approach $\mathscr{L}$ is defined over $TM\backslash 0$. In this case the Finsler Lagrangian is {\em reversible} if $\mathscr{L}(x,-v)=\mathscr{L}(x,v)$.  Now  a selection has to be made of {\em future} timelike cone (for reversible Lagrangians there is always a time oriented double covering). It is known \cite{beem70,beem74,perlick06,minguzzi13c} that the set of timelike vectors is made by the union of disjoint open sharp convex cones.
  In \cite{minguzzi13c} we proved that for reversible Lagrangians of Beem's type and for $n\ge 2$, there are indeed two timelike cones at each point, exactly as in Lorentzian geometry \footnote{John Beem investigated this problem and believed to have found a counterexample \cite{beem74} which, under closer inspection, can be shown to be incorrect.}. The Finsler spacetime in Beem's sense is then a time orientable Lorentz-Finsler manifold.

  Of course, Beem's spacetime can be regarded as a particular case of the rough setting where $\Omega$ can be identified with the future cone $I^+$. The attractive features a Beem's approach stand on the $C^2$ differentiability of the Lagrangian at the boundary of $\Omega=I^+$. This makes it possible to define lightlike geodesics. Also the standard theory of Finsler connections \cite{matsumoto86,ingarden93, antonelli93,anastasiei96,minguzzi14c,szilasi14},  being based on the slit tangent bundle, is well adapted to Beem's framework.

One could try to improve the rough theory, while refraining from adopting Beem's approach, through the assumption of Beem's differentiability conditions at the boundary of the lightlike cone as done in \cite{javaloyes13}. However,  this strategy does not lead to a different physical theory since we prove

\begin{theorem} \label{odd}
Let $\mathscr{L}\colon I^+ \to \mathbb{R}$ be a Lorentz-Finsler Lagrangian in the rough sense, and suppose furthermore that
\begin{itemize}
  \item[$\dagger$:]  $\mathscr{L} \in C^k$, $2\le k\le \infty$, can be $C^k$ extended to $\p \Omega$ while preserving the signature of $g$ and in such a way that $\dd \mathscr{L}\ne 0$ on $\p\Omega$,
  \end{itemize}
then there is an extension $\mathscr{L}\colon TM\backslash 0 \to \mathbb{R}$ to a  $C^k$   Finsler Lagrangian, namely $\mathscr{L}$ is positive homogeneous of degree two and its fiber Hessian  has Lorentzian signature. Furthermore, the locus $\mathscr{L}< 0$ is at each point $x\in M$ the union of two open convex sharp cones whose closures do not intersect and whose union spans $T_xM$ (to be interpreted as the future and past timelike cones).
\end{theorem}

\begin{remark} \label{dih}
In the previous theorem $C^k$ differentiability over the closed set $\bar{\Omega}$ should be understood in Whitney's sense \cite{whitney34}. Then by  Whitney extension theorem $\mathscr{L}$ can be $C^k$ extended in a neighborhood of $\p I^+=E^+$. The condition $\dd \mathscr{L}\ne 0$ implies  that the boundary $E^+$ is a $C^k$ embedded manifold, and by continuity $\mathscr{L}$ has Lorentzian Hessian in a neighborhood of $E^+$. Theorem \ref{odd} states more, namely that for $k=2$ the Finsler metric can be continuously extended outside the cone all over $TM\backslash 0$ remaining Lorentzian; if $k=3$ the Cartan torsion can be continuously extended, if $k=4$ the Cartan curvature can be so extended, and so on. Of course these extensions are not unique since the extension of $\mathscr{L}$ is not unique. Observe that a necessary condition for the extension of the metric (Cartan torsion) is the $C^2$ (resp.\ $C^3$) differentiability of the light cone.
\end{remark}

 In conclusion, demanding the nice differentiability properties of Beem's theory at the lightlike cone leads to a Finsler Lagrangian which is the restriction of a Beem's Lagrangian to $I^+$. Thus one is forced either to work with the rough theory and so work out a different definition of lightlike geodesic, e.g.\ as the limit of timelike geodesics (see below), or essentially work with a Beem's Lagrangian, either restricted to a cone or not.
  This result is important for it proves that there is only one Finslerian gravity theory for sufficiently differentiable Lagrangians on the light cone.

We end the section summarizing recent progress on the geometrical aspects of Finsler spacetime theory.

In \cite{minguzzi13c} we proved that the Legendre map $v\to g_v(v,\cdot)$ is bijective along with two inequalities, announced  in \cite{minguzzi13d},  that had curiously passed unnoticed, namely the Lorentz-Finsler reverse Cauchy-Schwarz inequality and the Lorentz-Finsler  reverse triangle inequality. As it is there clarified these results hold also for the rough theory, since the proofs involve just the future cone.

In \cite[Sect.\ 1.4]{minguzzi13d} we showed   how to extend global causality theory to the Finslerian domain using Beem's definition. Indeed, we proved:
\begin{itemize}
 \item[(a)] existence of convex normal neighborhoods (Theor.\ 1.16), and
 \item[(b)] that curves that  move pointwisely slower than light move also locally slower than light (Theor.\ 1.23), in other words that lightlike geodesics are locally achronal.
\end{itemize}
From here standard causality theory follows \cite{beem96}
since proofs generalize word for word from the Lorentzian domain,
 as long as they do not involve tricky aspects related to curvature \cite{minguzzi13d}.
For instance, the chronological relation is open; two points connected by a  causal curve are either chronologically related or the curve is an achronal lightlike geodesic;  globally hyperbolic spacetimes are defined in the usual way and split as a product; the causal hierarchy of spacetimes holds unaltered, and so on.
The correspondence with the Lorentzian causality theory is so complete that it has  yet to be be found a significative instance where it fails. For instance, the notable singularity theorems of Lorentzian geometry generalize to the Finslerian framework \cite{aazami14,minguzzi15}.

In all these applications it is convenient to work with a Beem's Lagrangian if available. One main advantage is that in Beem's formulation there is a clear notion of convex neighborhood \cite{minguzzi13d}. While the existence of convex neighborhoods is not necessary for the development of the theory, it helps greatly since many Lorentzian proofs, being based on this concept, pass word for word to the Finslerian domain.
This work, by establishing the equivalence of the Beem's and the conic approaches, proves that  all the good properties of Beem's theory pass to the conic formulation.

\begin{remark}
It can be observed that the limit curve theorems \cite{beem96,minguzzi07c} can be applied safely in the rough Finslerian theory. One could rework and modify the Lorentzian proofs; but there is a simpler fast trick to see this fact. Any sharp convex cone is the intersection of a countable number of round (elliptic) cones. As a consequence, on any Finsler spacetime we can find Lorentzian metrics $g^k$ such that the future causal cone of  the Finsler spacetime at any point is the intersection of the future Lorentzian cones. Thus if a sequence of inextendible Finsler causal curves $\gamma_i$, accumulates to a point $p$, it  converges, by the Lorentzian limit curve theorem \cite{beem96,minguzzi07c}, to a limit curve $\gamma$ passing through $p$ which is $g^k$-casual for every $k$ and hence Finsler causal. Therefore, one can define Finsler lightlike geodesics in the rough theory as limits of timelike geodesics.
\end{remark}

\section{Extending the Finsler Lagrangian}

In order to obtain the extension we have to study some properties of  $C^2$ Finsler Lagrangians defined on the slit tangent bundle, so let $\mathscr{L}$ be defined all over $TM\backslash 0$. Let us define the subsets of   $TM\backslash 0$, of timelike vectors $I=\mathscr{L}^{-1}((-\infty,0))$, lightlike vectors $E=\mathscr{L}^{-1}(0)$ and causal vectors $J=I\cup E$.
Let $\mathscr{I}^+\subset (TM\backslash 0)\backslash J$ be the subbundle obtained through the condition $2\mathscr{L}=1$. This is called {\em spacetime indicatrix} and $\mathscr{I}^+_x=\mathscr{I}^+\cap T_xM$ is the spacetime indicatrix at $x$. By positive homogeneity the Finsler Lagrangian is determined outside $J$ by the spacetime indicatrix, indeed
\[
\mathscr{L}(x,v)=s^2/2, \ \textrm{where} \ v/s\in \mathscr{I}^+_x .
\]
Let us recall a few well known facts \cite{laugwitz11,beem70} on the relationship between Finsler metric at a given point $x\in M$ and second fundamental form (affine metric) of $\mathscr{I}^+_x\subset T_xM$.

The Liouville vector field on $L:T_xM\to TT_xM$ is given by  $v\in TT_xM$ at $v\in T_xM$ and is transverse to the indicatrix
because $\frac{\p \mathscr{L}}{\p v^\alpha} v^\alpha=2\mathscr{L}=1\ne 0$. Let $D$ denote the usual derivative on $T_xM$ due to its affine structure, and let $X$ and $Y$ be two vectors tangent to the indicatrix at $v\in \mathscr{I}^+_x$. Let us extend them in a neighborhood $U\subset T_xM\backslash 0$ of $v$ to two vector fields which are tangent to $\mathscr{I}^+_x$. Since $Y$ is tangent to the indicatrix we have over $U\cap \mathscr{I}^+_x$, $\frac{\p \mathscr{L}}{\p v^\alpha} Y^\alpha=0$. Thus we can define the affine metric $h$ (second fundamental form) in the sense of affine differential geometry \cite{nomizu94} with the Gauss equation for the minus-Liouville transverse field
\begin{align} \label{gau}
D_X Y&=\nabla_X Y - h(X,Y) v, \qquad (Gauss).
\end{align}
Here $\nabla$ is a covariant derivative induced over the indicatrix which is not necessarily the Levi-Civita connection of $h$ (the difference between the two connections is the cubic form \cite{nomizu94}).
Contracting with $\dd \mathscr{L}$ and using $\frac{\p \mathscr{L}}{\p v^\alpha} \, v^\alpha=2\mathscr{L}$ and
\[
\dd \mathscr{L}(D_X Y)= \frac{\p \mathscr{L}}{\p v^\alpha} (X^\beta \frac{\p Y^\alpha}{\p v^\beta})= X^\beta \frac{\p }{\p v^\beta} (\frac{\p \mathscr{L}}{\p v^\alpha} Y^\alpha)-X^\beta Y^\alpha  \frac{\p^2 \mathscr{L}}{\p v^\alpha \p v^\beta}=-g_v(X,Y)
\]
we arrive at
\begin{equation}
g_{v}(X,Y)=h(X,Y) 2 \mathscr{L}.
\end{equation}

Summarizing,  on the spacetime indicatrix $2\mathscr{L}=1$ we have:
\begin{proposition}
The   Finsler metric $g_v$ induces on the spacetime indicatrix $\mathscr{I}^+_x$ a metric which coincides with the affine metric of the indicatrix, where the indicatrix is  regarded as a hypersurface of the affine space $T_xM$ with centro-affine transverse field $-v$ ($v$ is the Liouville vector field).
\end{proposition}
Observe that the tangent space to the indicatrix at $v\in \mathscr{I}^+_x$ is $\textrm{ker} \, g_v(v,\cdot)$, and since the transverse field $v$ is $g_v$-spacelike, $g_v(v,v)=2\mathscr{L}(x,v)=1$, the affine metric on the spacetime indicatrix is Lorentzian.
Thus if we are given a Finsler Lagrangian $\mathscr{L}$ on a sharp cone $J^+$, its extension on the slit tangent bundle can be accomplished through the construction of a spacetime indicatrix with Lorentzian affine metric.

The next result is related to a known convenient parametrization of the indicatrix \cite{loftin02}.
\begin{theorem} \label{lpo}
Let ${\bf t}$ be  projective coordinates on $T_x M\cap\{v: v^0>0\}$ so that $v=(v^0,{\bf v})=-\frac{1}{u}\,(1,{\bf t})$. The spacetime indicatrix $\mathscr{I}^+_x$, regarded has an embedding $\sigma:{\bf t} \to -\frac{1}{u({\bf t})}\,(1,{\bf t})$ has affine metric relative to the transverse field $-v$ given by
\[
h= \frac{u_{ij}}{u}\, \dd t^i \dd t^j .
\]
\end{theorem}
\begin{proof}
Let us observe that
\[
\sigma_{*}(e_j)=D_{\sigma_{*}(e_j)} v=\p_j \{-\frac{1}{u({\bf t})}\,(1,{\bf t})\}=-\frac{u_j}{u} \, v-\frac{1}{u({\bf t})}\, \p_j(1,{\bf t}) ,
\]
where $\p_j$ is a shorthand for $\p/\p {t^j}$. Thus
\begin{align*}
D_{\sigma_{*}(e_i)} \sigma_{*}(e_j)&=D_{\sigma_{*}(e_i)}D_{\sigma_{*}(e_j)} v=\p_i \p_j\{-\frac{1}{u({\bf t})}\,(1,{\bf t})\}\\
&=\frac{u_{ij}}{u}\, (-v)+\frac{u_j}{u^2}\, \p_i(1,{\bf t})+\frac{u_i}{u^2}\, \p_j(1,{\bf t})+2\frac{u_iu_j}{u^2}\, v\\
&=-\frac{u_i}{u}\,\sigma_{*}(e_j)-\frac{u_j}{u}\,\sigma_{*}(e_i)+ \frac{u_{ij}}{u}\, (-v).
\end{align*}
The first two terms are tangent to the indicatrix, thus the last one gives the affine metric.
\end{proof}

\begin{remark}
Of course the coordinate $v^0$ does not play any privileged role here. For instance an analogous statement holds on the region $v^1>0$ for similarly introduced projective coordinates. This result tells us that the Lorentzianity of the affine metric of the indicatrix can be read from a certain Hessian related to the projective radius.
\end{remark}

The Lorentzianity of the affine metric on the indicatrix will be obtained passing through the convexity of the level sets of its graph. For this reason we shall need a lemma on convexity.
We recall that  a $C^2$ function is convex (strictly convex) iff the Hessian is positive semidefinite (resp.\ definite). A quasi-convex function is a function for which the sublevel sets are convex. With $B(p,r)\subset \mathbb{R}^n$ we denote the open ball centered at $p$ of radius $r$.

\begin{lemma} \label{mps}
Let $\Omega \subset \mathbb{R}^n$ be a  compact convex neighborhood of the origin  and let $A$ be a  closed (possibly empty) set such that $A\subset \textrm{Int}\,\Omega$. Let $f$ be a  $C^k$, $2\le k\le \infty$, function on $\Omega$ which is strictly convex on $\Omega -A$. Then there are constants $a,b,r>0$, $\Omega \subset B(0,r)$, such that  $f$ can be extended to a $C^k$   function on $\mathbb{R}^n$,  strictly convex on $\mathbb{R}^n-A$, in such a way that $f(x)=-a+b \Vert x\Vert^2$ for $\Vert x\Vert\ge r$.
\end{lemma}

The  proof is based on the elaboration of some nice ideas on the extension of convex $C^2$ functions due to Min Yan \cite{minyan13}.

\begin{proof}
By Whitney's extension theorem \cite[Theor.\ I]{whitney34}  $f$ can be extended to $\mathbb{R}^n$ preserving its differentiability properties. Let $r_1>0$ be such that $\Omega_1\subset B(0,r_1)$. According to a recent result by  Min Yan \cite[Theor.\ 4.4]{minyan13} $f$ can actually be extended to a $C^k$  function \footnote{The $C^k$, $k>2$, case is not mentioned in \cite{minyan13} but follows immediately from the proof.} on $\mathbb{R}^n$, strictly convex on  $\mathbb{R}^n-A$. From now on we denote with $f$ this extension.

Let $0<\epsilon<r_1$ and let $\alpha\colon \mathbb{R}^n \to [0,1]$ be a smooth function such that
$\alpha=1$ on $\bar{B}(0,2r_1)$ and $\alpha=0$ outside $B(0,2r_1+\epsilon)$.
Let  $\beta\colon [0,+\infty) \to [0,1]$ be a smooth increasing function such that $\beta(s)=0$, for $s\in [0,2r_1-\epsilon]$, $\beta(s), \beta'(s)>0$ for $2r_1-\epsilon<s<3r_1$, and $\beta =1$ in a neighborhood of $[4r_1,\infty)$. Let $\gamma\colon [0,+\infty) \to \mathbb{R}$ be defined by
\[
\gamma(t)=\int_0^t s \beta(s)\, d s.
\]
Finally, let
\[
g(x)=\gamma(\Vert x\Vert),
\]
then
\[
h=\alpha f+ 2b g,
\]
is the searched $C^k$  extension, where $b>0$ is a  sufficiently large constant.
It is an extension because  $g=0$ on $\bar{B}(0,2r_1-\epsilon)$ where $\alpha=1$. For $v\in T_x\mathbb{R}^n$
\[
Hg(v,v)=\frac{1}{\Vert x\Vert}\Big\{ \gamma'(\Vert x\Vert) \Vert v\Vert^2+\beta'(\Vert x\Vert)(x\cdot v)^2\Big\},
\]
thus the Hessian of $g$ is positive definite outside $\bar{B}(0,2r_1-\epsilon)$ and positive semidefinite in $\bar{B}(0,2r_1-\epsilon)$. Now, $h$ has positive definite Hessian on $\bar{B}(0,2r_1)-A$ since  it coincides with $f+bg$ on $\bar{B}(0,2r_1)$. It has positive definite Hessian on the compact set $\bar{B}(0,2r_1+\epsilon)-{B}(0,2r_1)$ for sufficiently large $b$, since the Hessian of $\alpha f$ is bounded from below there. Finally, it has positive Hessian outside $\bar{B}(0,2r_1+\epsilon)$ since it coincides with $bg$ there.
Observe that $\gamma(t)=\frac{1}{2}t^2-\int_0^t s (1-\beta(s))\, d s$ where the last integral is non-negative. Thus since $\beta=1$ in a neighborhood of $[4r_1,+\infty)$ we have that $\gamma(t)=\frac{1}{2}t^2-a/(2b)$ there for some constant $a>0$. As a consequence, $h=b \Vert x\Vert^2-a$ for $\Vert x\Vert\ge r:= 4r_1$ and renaming $h\to f$ the theorem is proved.
\end{proof}

We are ready to prove our equivalence theorem.
 We have not been able to accomplish reversibility of the extended Lagrangian so $J^-$ might be different from $-J^+$.

\begin{proof}[Proof of theorem \ref{odd}]
Let $u\in I^+=\textrm{Int} J^+$ be a future-directed timelike vector field, and let $N\subset TM$ be a linear codimension one subbundle of the tangent bundle such that $N_x$ does not intersect $J_x$. The affine space $u+N_x$ intersects $J^+$ in a convex set $C_x$ with $C^k$  boundary. Let us introduce an auxiliary Riemannian metric on $M$ in such a way that $u$ has unit norm (recall that $M$ admits a Riemannian metric $\tilde{\gamma}$, see \cite{nomizu61}, thus $\gamma=\tilde \gamma/\tilde\gamma(u,u)$ would do). Thus we can regard $T_xM$ as a Euclidean space where the scalar product is given by $\gamma_x$. This scalar product will be denoted with a dot and we shall denote the induced norm with $\Vert \, \Vert$. From now on we shall focus on a single tangent space $T_xM$ and on its spacetime indicatrix. In order to get the extension over the whole manifold several constants introduced in the following argument should be regarded as smooth functions of $x$ although the dependence on $x$ will not be mentioned.

Let us introduce coordinates $(v^0,v^1,\cdots,v^{n})$ over  $T_xM$ so that $u=\p_0$, $N_x=span\{\p_1,$ $\cdots,\p_{n}\}$, and the basis $\{\p_1,\cdots,\p_{n}\}$ of $N_x$ is orthonormal. In this way $u+N_x=(v^0)^{-1}(1)$. Let $\tilde{\mathscr{L}}=\mathscr{L}\vert_{v^0=1}$ so that
\begin{equation} \label{mox}
\mathscr{L}(v)=(v^0)^2 \tilde{\mathscr{L}}({\bf v}/v^0).
\end{equation}
Since $\mathscr{L}$ is defined so far only over $J_x$, $\tilde{\mathscr{L}}$ is defined so far only over $C_x$ but we are now going to extend it over $N_x$.

We know (Remark \ref{dih}) that $\mathscr{L}$ can be extended over a cone neighborhood of $J^+$ preserving the Lorentzianity of the Hessian, so  $\tilde{\mathscr{L}}$ can be also extended in a neighborhood of $C_x$.

We shall first show that $\tilde{\mathscr{L}}$ can be redefined outside $C_x$ preserving its level sets, but so as to make it strictly convex. In this way we will be able to apply Lemma \ref{mps} and hence obtain a special quadratic behavior at large distances. This simple dependence will be important in order to construct $\mathscr{L}$ in the region $v^0\le 0$. Then in a last step we shall show how to redefine $\tilde{\mathscr{L}}$ in an intermediate region so as to get Lorentzianity of the Hessian of $\mathscr{L}$ everywhere.

If $v\in \p C_x$ then $\dd \mathscr{L}=g_v(v,\cdot)$ where $v$ is a lightlike vector, the ker of this one-form being  $T_vE^+_x$. Since $u+N_x$ is transverse to $E^+_x$, $\dd \tilde{\mathscr{L}}\ne 0$ on $\p C_x$.
Since $E^+_x$ is a null hypersurface in the spacetime $(T_xM,g_v)$ its induced space metric must be positive definite, thus if $X\in T_v\p C_x$, $X\ne 0$, as $X \not\propto v$, we have
\[
0<g_v(X,X)=H{\mathscr{L}}(X,X)=H\tilde{\mathscr{L}}(X,X).
\]
In other words the Hessian of $\tilde{\mathscr{L}}$ is positive definite when restricted to the tangent space to $\p C_x$. Let $0<\varepsilon<1$ be so
small that $\mathscr{L}$ has Lorentzian Hessian over $C_x^{3\varepsilon}:=\tilde{\mathscr{L}}^{-1}((-\infty,3\varepsilon)])$, and $\tilde{\mathscr{L}}$ has strictly convex compact sublevel sets. In particular,  $\tilde{\mathscr{L}}$ is quasi-convex on $C_x^{3\varepsilon}$. Let $h\colon [2\varepsilon, 3 \varepsilon]\to \mathbb{R}$ be a smooth increasing convex function
dependent on a parameter $c>0$ such that it coincides with $x$ in a neighborhood of $2\varepsilon$, it has first derivative at $3\epsilon$ bounded by a constant independent of $c$ and it has second derivative at $3\epsilon$ equal to $c$.
Then for sufficiently large $c>0$, $h(\tilde{\mathscr{L}})$ is strictly convex in a relatively compact inner neighborhood $C_x^{3\varepsilon}-A$, $A\supseteq C_x^{2\varepsilon}$, of $\p C_x^{3\varepsilon}$, while it remains quasi-convex on $C_x^{3\varepsilon}$.

Applying Lemma \ref{mps} with $\Omega=C_x^{3\varepsilon}$ we get that $\tilde{\mathscr{L}}$ on $C_x^{2\varepsilon}$ can be extended over $u+N_x$ to a quasi-convex function which coincides with $b {\bf v}^2-a$ for sufficiently large $\Vert{\bf v}\Vert$. Let ${\bf t}$ be projective coordinates on $T_xM$ so that $v=(v^0,{\bf v})=v^0(1, {\bf t})$. The graphing function $v^0({\bf t})$ determining the indicatrix $2\mathscr{L}=1$ on the region $v^0>1$ is obtained from Eq.\ (\ref{mox}), $v^0({\bf t})=1/\sqrt{2\tilde{\mathscr{L}}({\bf t})}$.
In the notations of Theorem \ref{lpo} we have  $u=-\sqrt{2\tilde{\mathscr{L}}}$ thus we have only to show that $\sqrt{\tilde{\mathscr{L}}}$ has  Lorentzian Hessian.

For some $r>0$, and $\Vert {\bf t}\Vert>r$ this is so because $\tilde{\mathscr{L}}=b {\bf t}^2-a$ (or because $\mathscr{L}=-a (v^0)^2+b {\bf v}^2$). However we shall need to change $\tilde{\mathscr{L}}$ in an intermediate region.

For ${\bf t}\in C_x^{2\varepsilon}-C_x$  Lorentzianity follows from the fact that ${\mathscr{L}}$ is Lorentzian on the corresponding convex cone neighborhood of $J_x^+$. Observe that $\p C^\delta$, $0\le \delta\le 2\varepsilon$, is the boundary of a strictly convex set, thus as they are level sets for $\sqrt{\tilde{\mathscr{L}}}$, they are spacelike hypersurfaces on $u+N_x$ for the metric given by the Hessian of $\sqrt{\tilde{\mathscr{L}}}$, and so $\sqrt{\tilde{\mathscr{L}}}$ plays the role of a time function there.

As $\tilde{\mathscr{L}}$ is increasing and convex outside $C^{3\varepsilon}_x$ it goes to infinity on every radial line of $u+N_x$ starting from ${\bf t}=0$. We wish to prove the Lorentzianity of $H\sqrt{\tilde{\mathscr{L}}}$ on a compact set $\tilde{\mathscr{L}}^{-1}([\varepsilon^2,R^2])\supset \{{\bf t}: \Vert {\bf t}\Vert\le r\}-C^{2\varepsilon}_x$. The idea is to modify $\sqrt{\tilde{\mathscr{L}}}$ through composition with an increasing function $F$ so as to preserve its convex sublevel sets but adjust its Hessian.
Let $X$ be a vector field on $(u+N_x)\backslash C_x$ transverse to $\p C^\sigma_x$, $\sigma\ge \epsilon$, such that $\p_X\tilde{\mathscr{L}}>0$, $X$ is timelike with respect to $H\sqrt{\tilde{\mathscr{L}}}$ on $C^{2\varepsilon}_x- C_x$  and $X$ coincides with ${\bf t}$  for $\Vert {\bf t} \Vert \ge r$. Let
\[
\beta=\sup_{\tilde{\mathscr{L}}^{-1}([\varepsilon^2, R^2])} \frac{H\sqrt{\tilde{\mathscr{L}}}(X,X)}{\big(\p_X \sqrt{\tilde{\mathscr{L}}}\,\big)^2}.
\]
Let $\alpha\colon [0,+\infty)\to [0,+\infty)$ be a smooth function which vanishes in $[0,\varepsilon/2]$ and in a neighborhood of $[2R,+\infty)$, and which is larger than $\beta$ on $[\varepsilon, R]$. Let
\[
F(x)= \int_0^x e^{-\int_0^y \alpha(z) \dd z} \dd y.
\]
Note that $F'>0$, $F=x$ on $[0,\varepsilon/2]$, while $F=A+B x$ with $A,B>0$ in a neighborhood of $[2R,+\infty)$. The inequality $A>0$ is easily checked from
\[
\int_0^{2R} \big(e^{\int_0^{2R}\alpha \dd z-\int_0^{y}\alpha \dd z}-1\big) \dd y>0 \Rightarrow \int^{2R}_0 e^{-\int_0^{y}\alpha \dd z} \dd y>2R e^{-\int_0^{2R}\alpha \dd z}\Rightarrow A>0.
\]
The equality $F=x$ assures that $F\circ \sqrt{\tilde{\mathscr{L}}}$ coincides with $\sqrt{\tilde{\mathscr{L}}}$ in a neighborhood of $C_x$.
Observe that
\[
H\big(F\circ \sqrt{\tilde{\mathscr{L}}}\,\big)(Y,Y)=F' H\sqrt{\tilde{\mathscr{L}}}(Y,Y)+F'' \big(\p_Y \sqrt{\tilde{\mathscr{L}}}\,\big)^2,
\]
so when $Y$ is tangent to the level sets of $F\circ \sqrt{\tilde{\mathscr{L}}}$ (or $\tilde{\mathscr{L}}$) we have by their strict convexity $H\big(F\circ \sqrt{\tilde{\mathscr{L}}}\, \big)(Y,Y)>0$, while when $Y=X$
\[
H\big(F\circ \sqrt{\tilde{\mathscr{L}}}\, \big)(X,X)=F' \bigg\{\frac{H\sqrt{\tilde{\mathscr{L}}}(X,X)}{(\p_X\sqrt{\tilde{\mathscr{L}}})^2}-\alpha\bigg\}\big(\p_X\sqrt{\tilde{\mathscr{L}}}\, \big)^2.
\]
Thus the Hessian is Lorentzian on the whole region $\mathscr{L}\ge 0$ (notice that $H\sqrt{\tilde{\mathscr{L}}}(X,X)<0$ in $C^{2 \varepsilon}_x- C_x$). Redefining $F\circ \sqrt{\tilde{\mathscr{L}}}\to \sqrt{\tilde{\mathscr{L}}}$ we get Lorentzianity of $\mathscr{L}$ all over the half-space $v^0>0$.

Observe that outside a cone containing the positive $v^0$-axis $\mathscr{L}$ is rotationally symmetric around that axis. As its level sets are round,  in order to infer the Lorentzianity of the indicatrix it is sufficient to check that the intersection between the indicatrix and the $v^0-v^1$ plane is strictly convex when expressed as a function $v^1(v^0)$.
Outside the mentioned cone
\[
\mathscr{L}=(v^0)^2\Big(A+B\sqrt{-a+b ({\bf v}/v^0)^2}\, \Big)^2 ,
\]
thus
\[
v^1(v^0)=\Big[\frac{a}{b}+\frac{1}{B^2b}\Big(\frac{1}{\sqrt{2}}-A v^0\Big)^{\!2}\Big]^{1/2}, \quad v^0>0.
\]
This function is convex for $v^0>0$ but it does not match with that obtained for $v^0<0$ through reflection. However, since $\p v^1/\p v^0\vert_{v^0\to 0}<0$ it is possible to continue the function for $v^0\le 0$ (and modify it slightly near $v^0=0$ if needed) while preserving strict convexity so as to match $v^1=\sqrt{d+e(v^0)^2}$, $d>0$, for some sufficiently large $e>0$, on $v^0 <0$ for sufficiently large $\vert v^0\vert$. In this way $\mathscr{L}$ has Lorentzian Hessian everywhere and coincides with a quadratic function $\mathscr{L}=[-e(v^0)^2+{\bf v}^2]/(2d)$ in a neighborhood of the past causal cone. The other statement of the theorem are immediate from the construction or follow from the general results of \cite{minguzzi13c}.
\end{proof}
\section{Conclusion}

We have shown that in sufficiently regular cone Lagrangian theories one can first pass to a slit tangent space Lagrangian, safely use the concept of convex neighborhood and space geodesic, and so translate almost effortlessly several Lorentzian results into their Finsler-Lorentzian analog.
One can also use the standard theory of Finsler connections over the slit bundle $TM\backslash 0$. This result proves the essential uniqueness of sufficiently regular Finsler spacetime theories.

Essentially, one is really forced to work on the cone only if the theory is sufficiently rough at the boundary of the cone, in which case one should be ready to adopt a rough Finslerian theory for the lightlike geodesics. In  this case there is one direction which is particularly worth exploring.  I refer to rough Finsler spacetimes which admit a vanishing mean Cartan torsion, or equivalently,  affine sphere indicatrices \cite{minguzzi14c}. They will be duly investigated in a different work.

\section*{Acknowledgments}   This work has been
partially supported by GNFM of INDAM.  The preprint version of this work appeared in arXiv:1412.4228.


\end{document}